\documentclass{llncs}

\usepackage{amsmath}
\usepackage{amssymb}
\usepackage{latexsym}

\usepackage{cite}
\usepackage{url}

\usepackage[algosection,vlined,linesnumbered,oldcommands]{algorithm2e}

\DeclareMathOperator{\dist}{dist}
\DeclareMathOperator{\parent}{par}

\begin{document}

\title{Pseudo-scheduling: A New Approach to the Broadcast Scheduling Problem}

\author{Shaun N.~Joseph\inst{1} \and Lisa C.~DiPippo\inst{2}}

\institute{mZeal Communications\thanks{Research done while at the University of Rhode Island.}, Littleton MA 01460 USA\\
  \email{shaun.joseph@mzeal.com}
  \and
  University of Rhode Island, Kingston RI 02881 USA\\
  \email{dipippo@cs.uri.edu}}

\maketitle

\begin{abstract}
  The broadcast scheduling problem asks how a multihop network of
  broadcast transceivers operating on a shared medium may share the medium
  in such a way that communication over the entire network is possible.
  This can be naturally modeled as a graph coloring problem
  via distance-2 coloring ($L(1,1)$-labeling, strict scheduling).
  This coloring is difficult to compute and may require a number of colors
  quadratic in the graph degree.
  This paper introduces \emph{pseudo-scheduling}, a relaxation of
  distance-2 coloring.
  Centralized and decentralized algorithms that compute pseudo-schedules
  with colors linear in the graph degree are given and proved.
\end{abstract}


\section{Introduction}
\label{sec:intro}

The broadcast scheduling problem asks how an arbitrary multihop network of
broadcast transceivers operating on a shared medium may share the medium
in such a way that communication over the entire network is possible.
In particular, two or more transmissions made simultaneously
(in time and space) on the same medium should be expected to fail;
ie, the transmissions conflict.

A medium access control (MAC) protocol is a practical solution to the
broadcast scheduling problem.
The predominant approach to MAC protocol design is contention,
the outstanding example of which is carrier sense multiple access (CSMA).
Examples include the
wireless Ethernet standard 802.11 and the
protocol B-MAC~\cite{bmac} for wireless sensor networks.


The alternative to contention is explicit scheduling,
such as time-division multiple access (TDMA) or
frequency-division multiple access (FDMA).
Regardless of how the medium is divided, however, the allocation of quanta
to network nodes is naturally expressed as graph coloring problem; eg,
a graph coloring with ten colors might correspond to a TDMA frame with
ten timeslots.

There are a variety of graph coloring problems extant, but
the obvious and canonical model is the \emph{$L(1,1)$-labeling}, also
known as \emph{distance-2 coloring}, \emph{coloring of the graph square}, or
\emph{strict scheduling}.
Here a vertex must be colored differently from any other vertex at distance
one or two. The seminal results on this coloring were obtained by
McCormick~\cite{mccormick}, who found
that strict scheduling is NP-Complete as a decision problem;
and that the number of colors required is $\Delta^2+1$ in the
worst case, where $\Delta$ is the graph degree.

This paper introduces \emph{pseudo-scheduling},
a relaxation of strict scheduling.
Whereas strict schedules guarantee that every path in the graph is a
valid communication path, pseudo-schedules only require the existence of
some communication path between any two vertices;
the requisite paths may exist along the edges of a spanning tree,
in exact analogy to a network routing tree.


Pseudo-scheduling is defined formally as a graph coloring problem below.
In \S\ref{sec:related} we survey related work.
A centralized pseudo-scheduling algorithm using colors at most twice the
graph degree is presented in \S\ref{sec:2delta}, and in \S\ref{sec:dband}
we examine an algorithm that is decentralized but still uses colors only
linear in the graph degree (with a reasonable multiplicative factor).
We conclude in \S\ref{sec:conc}.

\subsection{Definitions}
\label{sec:intro:def}

Let us consider (vertex) coloring from the perspective of how colored vertices
do or do not ``conflict.''
Let $G = (V,E)$ be a simple, undirected graph with a coloring
$l : V \to \mathbb{Z}$.
We say that the ordered pair $(u,v) \in V^2$ is \emph{nonconflicting} iff
$uv \in E$, $l(u) \neq l(v)$, and
for all $x \neq u$ adjacent to $v$, $l(u) \neq l(x)$.
A directed path from $u$ to $v$ is likewise nonconflicting iff the pairs
that comprise it are nonconflicting.

Observe that a strict schedule can be defined as a coloring such that every
path in the graph is nonconflicting. Immediately we conceive of a natural
relaxation: instead of requiring that every path be nonconflicting, demand only
the existence of at least one nonconflicting path from $u$ to $v$ for every
$u,v$ that are connected in $G$.
Such a coloring we call a \emph{pseudo-schedule}.

It is convenient to work with edges $uv$ such that
both $(u,v)$ and $(v,u)$ are nonconflicting (under some coloring);
such an edge is said to be \emph{bidirectional}. A subgraph is bidirectional
iff its edges are bidirectional. A pseudo-schedule with a bidirectional
subgraph $H$ is an \emph{$H$-pseudo-schedule}.

A (symmetric link) network with a routing tree can be represented by a graph $G$
with spanning tree $T$. A $T$-pseudo-schedule $s$ of $G$ is then a very
interesting structure, as it ensures nonconflict along the routing tree
but allows it elsewhere.

For the remainder of this paper, we assume that all graphs are simple;
they may also be taken as connected without loss of generality.
$V(G)$ and $E(G)$ denote the vertex and edge sets of $G$, respectively.
$\Delta_G$ denotes the degree of graph $G$,
$\deg_G(v)$ the degree of a vertex in $G$,
and $N_G[v]$ the closed neighborhood,
that is, the set comprised of $v$ and all vertices adjacent to $v$ in $G$.
Finally, $\dist_G(u,v)$ is the graph distance between vertices $u,v$ in $G$.

\section{Related Work}
\label{sec:related}

Although contention remains the most common approach to solving the
broadcast scheduling problem, a number of explicit scheduling
algorithms and protocols have been proposed.
These can be separated into two categories: node-oriented scheduling,
which is naturally modeled by vertex coloring, and link-oriented scheduling,
modeled by edge coloring. We consider only the former, in particular
because we wish to exploit one-to-many broadcast transmissions.


Node-oriented scheduling has been held back by its approach to conflict
avoidance, which has hitherto taken strict scheduling as its starting point.
This is explicit in DRAND~\cite{drand},
which implements the greedy algorithm for strict scheduling.
Alternately, conflicts are tolerated in
Z-MAC~\cite{zmac} and Funneling-MAC~\cite{funneling_mac}, but only inasmuch
as these protocols combine TDMA with CSMA; the TDMA part of the protocol is
strict.
TSMA~\cite{tsma} and RIMAC~\cite{rimac} also permit conflicts, but this
is aimed primarily at making the schedule easier to compute via decentralized,
probabilistic methods, rather than reducing the division of the medium;
indeed, the division tends to increase.

In his work on RAC-CT~\cite{ren}, Ren exhibits and implements what is
basically the greedy algorithm for pseudo-schedules.
Although Ren gives empirical evidence that the algorithm uses a number of
colors very close to the graph degree on random grid graphs, the general
upper bound is quadratic in the graph degree, giving no asymptotic improvement
over strict scheduling.

\section{Twice-Degree Algorithm}
\label{sec:2delta}

If conditions permit the use of a centralized algorithm,
and we are more or less indifferent to the choice of spanning tree,
the algorithm presented in this section will produce a
pseudo-schedule of any graph $G$ in no more than $2\Delta_G$ colors.
The user chooses a root vertex $r$, most likely corresponding to a base
station/access point; the algorithm selects the spanning (routing) tree to its
convenience, although this tree will minimize distances to $r$.

For any tree $T$ rooted at $r$, we say that $u$ is the
\emph{parent} of $v$ and $v$ is a \emph{child} of $u$ iff $uv \in E(T)$ and
$\dist_T(u,r) < \dist_T(v,r)$; we write $u = \parent_{T,r}(v)$.
(It will be convenient to let $\parent_{T,r}(r)=r$.)
We also henceforth permit ourselves the following abuse of
notation: given a coloring $s$, which may be only partially defined,
for any set $U$ of vertices let $s(U)$ denote the set
$\{ s(u) ; u \in U , s(u) \textrm{ is defined}\}$.

\begin{algorithm}
  \dontprintsemicolon
  \KwIn{$G$, a graph; and $r$, a distinguished vertex of $G$.}
  \KwOut{A pseudo-schedule on $G$.}
  $V_T \gets \{r\} , E_T \gets \emptyset$\;
  $T \gets (V_T,E_T)$\;
  $Q \gets \{r\} ,  Q' \gets \emptyset$ \tcp*[h]{$Q,Q'$ are queues}\;
  $s \gets \emptyset$\;
  \Repeat{$Q = \emptyset$}{
    \ForEach{$v \in Q$ (FIFO)}{
      $N \gets N_G[v] - V_T$\;
      append $N$ to $Q'$ (in any order)\;
      $V_T \gets V_T \cup N$\;
      $E_T \gets E_T \cup \{ vx ; x \in N \}$\;\nllabel{algo:2delta:adoption}
      $K \gets s(N_G[\parent_{T,r}(v)])$\;\nllabel{algo:2delta:parent-colors}
      \ForEach{$x \in N_G[v]$}{
        \If{$x \neq v$ and $vx \notin E_T$}{
          add $s(\parent_{T,r}(x))$ to $K$\;
        }
      }
      $k \gets \min (\mathbb{Z}^+ - K)$\;
      add $v \mapsto k$ to $s$\;
    }
    $Q \gets Q' , Q' \gets \emptyset$\;
  }
  \KwRet{$s$}\;
  \caption{The twice-degree algorithm}
  \label{algo:2delta}
\end{algorithm}



\subsection{Analysis}
\label{sec:2delta:analysis}


It is fairly obvious that the algorithm will use at most $2 \Delta_G$ colors;
the only difficulty is to show that it actually produces a pseudo-schedule.

\begin{theorem}
  Let $s$ be the coloring produced by the twice-degree algorithm
  with input $(G,r)$; then $s$ is a $T$-pseudo-schedule,
  where $T$ is the tree generated internally by the algorithm.
\end{theorem}

\begin{proof}
  Observe that $T$ is produced by a breadth-first search process and that
  every vertex at $r$-distance $i$ is colored before any vertex at $r$-distance
  $i+1$. We can also see that if vertices $u,v$ have distinct parents $p_u,p_v$
  respectively, then if $p_u$ was colored before $p_v$, $u$ was colored before
  $v$.

  Take $uv \in E(T)$ such that $u = \parent_{T,r}(v)$.
  First we show that $(u,v)$ is nonconflicting. $s(u) \neq s(v)$ clearly.
  Consider next any child $x$ of $v$; since $u$ is adjacent to
  $v = \parent_{T,r}(x)$, we have $s(u) \neq s(x)$.
  The only vertices left to check are those in $N_G[v] - N_T[v]$;
  let $y$ be such a vertex. Now if $\dist_G(r,u) < \dist_G(r,y)$, then $y$
  was colored after $u$, so $s(u) \neq s(y)$.
  If, on the other hand, $\dist_G(r,u) = \dist_G(r,y)$, then $u$ ``adopted''
  $v$ before $y$ could (line~\ref{algo:2delta:adoption}), which implies that
  $u$ was colored before $y$, hence $s(y) \neq s(x)$.

  Let us now establish that $(v,u)$ is nonconflicting. Obviously
  $s(v) \neq s(x)$ for any $x \in N_G[u]$ with $\dist_G(r,x) < \dist_G(r,v)$
  since in this case $x$ must have been colored before $v$.
  Turning to $x \in N_G[u]$ with $\dist_G(r,x) = \dist_G(r,v)$,
  clearly $s(v) \neq s(x)$ if $x$ is a child of $u$, since $v$ will be checked
  before coloring $x$ and vice-versa.
  If, on the other hand, $\parent_{T,r}(x)=w \neq u$, it must be that
  $w$ was colored before $u$ since the former adopted $x$, thus $s(x)$
  was defined when the algorithm computed $s(N_G[u])$
  (line~\ref{algo:2delta:parent-colors}) before coloring $v$, and
  $s(v) \neq s(x)$. \qed
\end{proof}




\section{d-Band Algorithm}
\label{sec:dband}

The twice-degree algorithm employs a very small number of colors in the worst
case, but the requirements for central control and user indifference to the
resulting spanning tree may not be reasonable in network applications.
The $d$-band algorithm does away with these requirements, albeit at the cost
of raising the worst-case number of colors used, although this remains linear
in the graph degree (with a reasonable coefficient).

Let $G$ be a graph with a spanning tree $T$ rooted at $r$.
Intuitively, the algorithm divides the graph into $d$ bands based on vertices'
$T$-distance from $r$ modulo $d$,
with each band being colored from its own palette, disjoint from every other.
The idea is that, if $d$ is sufficiently large,
we can rule out many conflicts \emph{a priori},
greatly reducing the number of vertices that have to be checked.

The $d$-band algorithm is decentralized, with each vertex
acting as an autonomous agent passing the following messages:
\begin{itemize}
\item REQ-COL($L$), where
  $L$ is a set of excluded colors;
\item PUT-COL($x$,$k$), where
  $x$ is a vertex being assigned color $k$;
\item RPT-COL($k$,$w$), where
  $k$ is the sender's color (if known) and
  $w$ is a vertex that must be colored before the sender is colored;
  or, in a ``reverse report,'' where $k$ is a color excluded for the sender;
\item RPT-PAR($k$,$w$), where
  $k$ is the color of the sender's parent (if known) and
  $w$ is a vertex that must be colored before the sender's parent is colored;
  or, in a ``reverse report,'' where $k$ is the color of $w$,
  a stepparent of the sender;
\item DEP-REQ($w$), where
  $w$ is a vertex whose color must be assigned before the sender can
  issue REQ-COL; and
\item DEP-PUT($w$), where
  $w$ is a vertex whose color must be assigned before the sender can
  issue PUT-COL.
\end{itemize}
The sending vertex is implicitly included in any message, along with
information about the intended receiver.
We let $\infty$ denote an unknown color and let
\begin{equation}
  \mathrm{Palette}(v) = \{ \dist_T(r,v) \bmod d + id ; i \geq 0\}
  \textrm{.}
\end{equation}

Along with the definitions of parent and child as in \S\ref{sec:2delta},
we say also that $u$ is a \emph{stepparent} of $v$ iff
$\dist_T(u,r) = \dist_T(v,r)-1$ and $uv \in E(G)-E(T)$;
and that $x$ is a \emph{stepchild} of $y$ iff
$\dist_T(x,r) = \dist_T(y,r)+1$ and $xy \in E(G)-E(T)$.
Each vertex is assumed to know its parent, children, stepparents, and
stepchildren. Additionally, each vertex knows its $T$-distance from $r$.
Finally, we assume that $V(G)$ admits a strict total order
$\prec$ that can be efficiently computed at any vertex.

The flow of the algorithm about a vertex $v$ can be sketched
roughly as follows:
\begin{enumerate}
\item $v$ listens for RPT-PARs from all of its stepchildren,
  building a list of excluded colors $L$.
\item $v$ sends REQ-COL($L$) to its parent $u$.
\item $u$ listens for RPT-COLs from all of its stepchildren,
  building a list of forbidden colors $K$.
\item $u$ sends PUT-COL($v$,$k_v$) to $v$ (and all stepchildren of $u$),
  where $k_v$ is the smallest color in the palette of $v$ not in $K \cup L$.
\item $v$ broadcasts RPT-COL($k_v$,$v$).
\item Each child of $v$ sends RPT-PAR($k_v$,$v$) to all its stepparents.
\end{enumerate}
(In this sketch, for the sake of simplicity we have ignored
the DEP facility.)
In general, the $d$-band algorithm colors the leaves of $T$ first and
proceeds towards the root, although significant parallelism is possible.


The root vertex $r$ assigns itself the color 0, making this known by sending
RPT-COL(0,$r$) to its children.
Any vertex besides $r$ must acquire its color as per AcquireColor
(Algorithm~\ref{algo:dband-acquire}).
Any vertex with children must assign colors to its children as per
AssignColors
(Algorithm~\ref{algo:dband-assign}).
Finally, any non-root vertex must receive and relay reports as per
ReportColors
(Algorithm~\ref{algo:dband-report}).
The ensemble of these procedures, running independently and in parallel on
every vertex simultaneously, constitutes the $d$-band algorithm.

(We assume fully reliable transmission with synchronous communication
primitives \emph{send} and \emph{listen}. AssignColors uses the primitive
\emph{ack-send $\mathcal{M}$ to $x$} by which is meant: send $\mathcal{M}$
to $x$ and wait until $\mathcal{M}$ is sent back as confirmation,
queuing any messages that arrive in the meanwhile for retrieval by the
next call to listen.)

\begin{algorithm}
  \dontprintsemicolon
  \KwIn{$v$, ``this'' vertex.}
  \KwOut{A color.}
  $L \gets \emptyset$\;
  $f_W \gets \{ x \mapsto v ; \textrm{$x$ is a stepchild of $v$}\}$\;
  $W \gets \mathrm{Image}(f_W)$\;
  $w_{\prec} \gets v$\;
  send RPT-COL($\infty$,$v$) to children of $v$\;
  \While{$W \neq \emptyset$}{\nllabel{algo:dband-acquire-phase1}
    listen for message $\mathcal{M}$\;
    \uIf{$\mathcal{M} = \textrm{RPT-PAR($k$,$w$)}$ from (step)child $x$ of $v$}{
      \uIf{$k \neq \infty$}{
        remove $x \mapsto f_W(x)$ from $f_W$\;
        add $k$ to $L$\;
      }
      \lElseIf{$w=v$}{remove $x \mapsto f_W(x)$ from $f_W$}\;\nllabel{algo:dband-acquire-cycle-break}
      \lElse{add/replace $x \mapsto w$ in $f_W$}\;
      $W \gets \mathrm{Image}(f_W)$\;
    }
    \ElseIf{$\mathcal{M} = \textrm{DEP-REQ($w$)}$ from a child of $v$}{
      \tcp{Reverse dependence}
      add/replace $w \mapsto w$ in $f_W$\;
    }

    \If{$w_{\prec} \neq  \min_{\prec} (W \cup \{v\})$}{
      $w_{\prec} \gets \min_{\prec} (W \cup \{v\})$\;
      send DEP-REQ($w_{\prec}$) to children of $v$\;
      send RPT-COL($\infty$,$w_{\prec}$) to stepparents of $v$\;\nllabel{algo:dband-acquire-ItoII}
    }
  }
  send RPT-COL($\infty$,$v$) to children of $v$\;
  send REQ-COL($L$) to the parent of $v$\;
  listen for PUT-COL($v$,$k_v$) from the parent of $v$\;
  send RPT-COL($k_v$,$v$) to children, stepchildren, and stepparents of $v$\;
  \KwRet{$k_v$}\;
  \caption{The d-band algorithm: AcquireColor}
  \label{algo:dband-acquire}
\end{algorithm}

\begin{algorithm}
  \dontprintsemicolon
  \KwIn{$v$, ``this'' vertex.}
  $K,f_L,f_W,W \gets \emptyset$\;
  $X \gets \textrm{stepchildren of $v$} , Z \gets \textrm{children of $v$}$\;
  $f_R \gets \{ z \mapsto \emptyset ; z \in Z \}$\;
  $w_{\prec} \gets v$\;
  \While{$Z \neq \emptyset$}{
    \uIf{$X = \emptyset$}{
      \ForEach{$z \mapsto L \in f_L$ such that $f_R(z) = \emptyset$}{
        $k_z \gets \min (\mathrm{Palette}(z) - K - L)$\;
        send PUT-COL($z$,$k_z$) to $z$ and stepchildren of $v$\;
        remove $z$ from $Z$\;
        add $k_z$ to $K$\;
      }
    }
    \ElseIf{$W \neq \emptyset$ and $w_{\prec} \neq \min_{\prec} W$}{
      $w_{\prec} \gets \min_{\prec} W$\;
      send DEP-PUT($w_{\prec}$) to children of $v$\;
    }
    listen for message $\mathcal{M}$\;
    \uIf{$\mathcal{M} = \textrm{REQ-COL($L$)}$ from child $z$ of $v$}{
      \lIf{$f_L(z)$ defined}{$L \gets L \cup f_L(z)$}\;
      add/replace $z \mapsto L$ in $f_L$\;
    }
    \uElseIf{$\mathcal{M} = \textrm{RPT-COL($k$,$w$)}$ from (step)child $x$}{
      \uIf{$x \in X$}{
        \uIf{$k \neq \infty$ or $w$ is a child of $v$}{
          remove $x$ from $X$\;
          remove $x \mapsto f_W(x)$ from $f_W$\;
          add $k$ to $K$\;
          \lIf{$w$ is a child of $v$}{ack-send DEP-PUT($v$) to $x$}\;
        }
        \lElse{add/replace $x \mapsto w$ to $f_W$}\;
        $W \gets \mathrm{Image}(f_W)$\;
      }
      \ElseIf{$x$ is a child of $v$}{
        \uIf{$k \neq \infty$}{
          $L \gets (\textrm{$f_L(x)$ defined ? } f_L(x) : \emptyset)$\;
          add $k$ to $L$\;
          add/replace $x \mapsto L$ in $f_L$\;
        }
        \lElse{remove $w$ from $f_R(z)$}\;
      }
    }
    \ElseIf{$\mathcal{M} = \textrm{DEP-PUT($w$)}$ from child $z$ of $v$}{
      \tcp{Reverse dependence}
      add $w$ to $f_R(z)$\;
    }
  }
  send PUT-COL($\infty$,$v$) to stepchildren of $v$\;
  \caption{The d-band algorithm: AssignColors}
  \label{algo:dband-assign}
\end{algorithm}

\begin{algorithm}
  \dontprintsemicolon
  \KwIn{$v$, ``this'' vertex.}
  $p \gets$ parent of $v$\;
  $W_{I},W_{II} \gets \emptyset$\;
  listen for message $\mathcal{M}$ from (step)parents\;
  \uIf{$\mathcal{M} = \textrm{DEP-REQ($w$)}$ from $p$}{
    \If{$w$ is a stepparent of $v$}{
      \tcp{Type I cycle breaking: reverse dependence}
      add $w$ to $W_I$\;
      send DEP-REQ($w$) back to $p$\;
    }
    send RPT-PAR($\infty$,$w$) to stepparents of $v$\;
  }
  \uElseIf{$\mathcal{M} = \textrm{RPT-COL($k$,$w$)}$ from $x$}{
    \uIf{$x=p$}{
      send RPT-PAR($k$,$w$) to stepparents of $v$\;\nllabel{algo:dband-report-IItoI-2}
    }
    \ElseIf{$x \in W_I$ and $k \neq \infty$}{
      \tcp{Type I cycle breaking: reverse report}
      remove $x$ from $W_I$\;
      send RPT-PAR($k$,$x$) to $p$\;
    }
  }
  \uElseIf{$\mathcal{M} = \textrm{DEP-PUT($w$)}$ from $x$}{
    \uIf{$x=p$}{
      send RPT-COL($\infty$,$w$) to children and stepparents of $v$\;\nllabel{algo:dband-report-IItoI-1}
    }
    \Else{
      \tcp{Type II cycle breaking: reverse dependence}
      add $x$ to $W_{II}$\;
      send DEP-PUT($x$) to $p$ and $x$\;
    }
  }
  \ElseIf{$\mathcal{M} = \textrm{PUT-COL($u$,$k$)}$ from $x \in W_{II}$}{
    \tcp{Type II cycle breaking: reverse report}
    \lIf{$k = \infty$}{remove $x$ from $W_{II}$}\;
    send RPT-COL($k$,$x$) to $p$\;
  }
  \caption{The d-band algorithm: ReportColors}
  \label{algo:dband-report}
\end{algorithm}

\subsection{Analysis}
\label{sec:dband:analysis}

We say that the $d$-band algorithm \emph{terminates} on a particular
graph with rooted spanning tree iff AcquireColor returns on every vertex.

\begin{theorem}
  \label{thm:dband-term}
  The $d$-band algorithm terminates on any graph with any choice of
  root and spanning tree.
\end{theorem}

\begin{proof}
  Let $G$ be a graph with spanning tree $T$ rooted at $r$.
  AcquireColor (Algorithm~\ref{algo:dband-acquire})
  on vertex $v$ does its main work in the loop beginning at
  line~\ref{algo:dband-acquire-phase1}.
  As this loop is bypassed when $v$ has no stepchildren, assume that it does.
  We say that $v$ has a \emph{request dependence} on $u$ when $u$ is the parent
  of a stepchild of $v$; and just as $v$ depends on $u$, $u$ may depend on $t$,
  and so on. If we can follow the dependency chain to some terminal $a$ that has
  no stepchildren, there is no problem,
  since we can inductively work back to $v$.
  However, the dependency chain may in fact be a cycle, in the sense that
  $v$ has request dependence on $u$, $u$ has request dependence on $t$, and so
  on up to $a$, but then $a$ has request dependence on $v$.
  This we call a \emph{dependency cycle of type I}.
  
  Given $C$, a dependency cycle of type I, let $v = \min_{\prec} C$.
  Assume, for the time being, that $C$ is the only dependency cycle in the
  graph.
  $v$ issues DEP-REQ($v$) to its children, and via
  ReportColors (Algorithm~\ref{algo:dband-report})
  one of the
  children sends RPT-PAR($\infty$,$v$) to $x$, which depends on $v$.
  But since $v \prec x$, $x$ issues DEP-REQ($v$) to its children,
  one of which then sends RPT-PAR($\infty$,$v$) to $y$, which depends on $x$,
  and so on.

  Let $u$ be the vertex in $C$ on which $v$ depends, creating the cycle.
  $u$ issues DEP-REQ($v$) to its children, and one of them transmits
  RPT-PAR($\infty$,$v$) to $v$. At this point $v$ can detect the dependency
  cycle, and $v$ breaks the cycle by ignoring its dependence on $u$
  (see line~\ref{algo:dband-acquire-cycle-break} of AcquireColor).
  As per our assumptions, $v$ is now free of request dependencies,
  or at worst sits in linear dependence chains that are naturally resolved;
  that is, $v$ (eventually) acts as if it has no stepchildren,
  and proceeds to issue REQ-COL to its parent $p$.
  
  Let us assume that $p$ eventually assigns a color to $v$ via PUT-COL.
  $v$ then broadcasts RPT-COL, resolving the now-linear dependency chain.
  (The resolution is a little unusual at $u$, where we have registered a
  ``reverse dependence'' on $v$---but this will be cleared by the RPT-COL
  broadcast from $v$, which causes a ``reverse report'' RPT-PAR to be sent
  to $u$ from one of its children.) Hence every vertex in $C$ gradually becomes
  free to issue REQ-COL, and if we assume that every one of their parents
  replies with PUT-COL, then AcquireColor terminates on every vertex in $C$.

  We now shift our attention to AssignColors
  (Algorithm~\ref{algo:dband-assign}).
  A vertex $v$ with parent $p_v$
  is said to have a \emph{put dependence} on any stepchild of $p_v$.
  (It is convenient for the dependence to be registered at $p_v$.)
  Just like request dependencies, put dependencies can be chained and may
  form a cycle; this we call a \emph{dependency cycle of type II}.
  It is not hard to see that a type II cycle is broken by essentially the
  same method used for the type I cycle,
  with DEP-PUT and RPT-COL standing in for DEP-REQ and RPT-PAR, respectively.
  (Once again, there is a special
  ``reverse dependence'' facility. Let $p_v$ have stepchild $u$ with parent
  $p_u$. If a type II cycle is broken at $p_v$, then $p_u$ will register
  the reverse dependence of $u$ on the children of $p_v$.
  Resolution comes when $p_u$ finishes coloring its children,
  with ``reverse report'' RPT-COLs being sent to $p_u$ via $u$.)
  Hence the $d$-band algorithm terminates in the presence of a
  dependency cycle of type II, provided that REQ-COL is issued.
  
  Finally, a dependency cycle of \emph{mixed type} is possible.
  AcquireColor handles the transition from request to put dependency by
  repackaging RPT-PAR as RPT-COL (line~\ref{algo:dband-acquire-ItoII}),
  while ReportColors handles the reverse transition by first repackaging
  DEP-PUT as RPT-COL (line~\ref{algo:dband-report-IItoI-1}) and then
  relaying the latter as RPT-PAR (line~\ref{algo:dband-report-IItoI-2}).
  But observe that a cycle of mixed type can be broken by the methods previously
  described; it is treated exactly as if it were a cycle of type I or type II if
  it is broken by AcquireColor or AssignColors, respectively. The same
  mechanisms then assure that the resolution proceeds across the cycle.
  We conclude that the $d$-band algorithm terminates if there is no more than
  one dependency cycle in the graph (of whatever type).


  Given a dependency cycle $C$ (of any type), observe that there exists $l$
  such that $\dist_T(v,r)=l$ for all $v \in C$; call $l$ the level of $C$.
  Clearly cycles with different levels cannot affect each other; additionally,
  disjoint cycles do not interact.
  Thus the $d$-band algorithm terminates given any number of disjoint dependency
  cycles per level of $T$.

  Unfortunately a graph may contain many overlapping dependency cycles;
  we claim the algorithm terminates regardless. Let $\mathcal{C}$ be
  a family of intersecting dependency cycles. As there is a strict total order
  $\prec$ on vertices, there exists some
  \begin{displaymath}
    v = \min_{\prec} \bigcup_{C \in \mathcal{C}} C
    \textrm{.}
  \end{displaymath}
  Observe that AcquireColor must terminate on $v$, since all dependency
  cycles in $\mathcal{C}$ containing $v$ will be broken at $v$, if not
  elsewhere. After breaking all such cycles and resolving all newly-linear
  chains, let $\mathcal{C}'$ be the remaining cycles.
  Obviously $|\mathcal{C}'| < |\mathcal{C}|$, and we can apply the same argument
  to $\mathcal{C}'$ inductively.
  We conclude that the $d$-band algorithm terminates. \qed
\end{proof}


Because the $d$-band algorithm terminates, we can treat the color output by
AcquireColor on every vertex as a coloring of the whole graph.
However, $d$ must be taken sufficiently large for coloring to be a
pseudo-schedule, as made precise by the following theorem, the proof of
which appears in \S3.4.1 of \cite{joseph-phd}.
Note that an immediate consequence of this is that $d=3$ suffices if $T$ is a
shortest path tree.

\begin{theorem}
  \label{thm:dband-dval}
  Let $G$ be a graph with spanning tree $T$ rooted at $r$.
  The $d$-band algorithm yields a $T$-pseudo-schedule provided that
  \begin{equation}
    d \geq \max_{uv \in E(G)} |\dist_T(u,r)-\dist_T(v,r)| + 2
  \end{equation}
  or $d$ is greater than the height of $T$.
\end{theorem}

We will consider the number of colors used by a $d$-band pseudo-schedule $s$
to be the greatest integer color appearing in $s$ plus one; this forces us
to account for ``gaps'' of unused colors.

\begin{theorem}
  Let $h$ denote the number of colors used by the $d$-band algorithm
  on graph $G$ with spanning tree $T$ rooted at $r$,
  where $d$ meets the conditions of Theorem~\ref{thm:dband-dval}.
  If $d \leq \mathrm{height}(T)+1$ and $\Delta_G \geq 2$, then
  \begin{equation}
    d(\Delta_T - 1) + 1 \leq h \leq 2d(\Delta_G-1)
  \end{equation}
  except possibly when $\Delta_T = 2$, in which case the lower bound
  falls to $d$.
\end{theorem}

\begin{proof}
  The lower bound for $\Delta_T = 2$ is established by any path graph.
  For $\Delta_T \geq 3$,
  let $v$ be the vertex on which $T$ achieves its maximum degree.
  Then any child of $v$ uses a palette containing at least the colors
  $a , a + d , \hdots , a + d(\Delta_T-1)$. With $a=0$, we obtain the lower
  bound.

  For the upper bound, consider a vertex $v$ with $\dist_T(v,r)=d-1$.
  Its parent $p_v$ has at most $\Delta_G-2$ neighbors
  distinct from $v$ but at the same $T$-level as $v$.
  Additionally, $v$ has at most $\Delta_G-1$ stepchildren,
  each of which could have a distinct parent.
  This is a total of $2\Delta_G-2$ vertices (including $v$), so
  \begin{displaymath}
    \mathrm{Palette}(v) \subseteq
    \{ d-1 , 2d-1 , \hdots , d-1 + d(2\Delta_G-3) \}
  \end{displaymath}
  which yields the upper bound. \qed
\end{proof}

\section{Conclusion}
\label{sec:conc}

This work has introduced and motivated pseudo-scheduling as a new approach to
the broadcast scheduling problem. The algorithms exhibited here prove
that pseudo-scheduling can result in asymptotic improvements in medium
division relative to strict scheduling.
This would correspond concretely, for instance, to TDMA frames that
grow only linearly with the neighborhood size, improving the network
throughput, especially in ad-hoc wireless networks in which the neighborhood
size cannot be tightly bounded in advance.
Although the concepts have yet to be
put to the test in practical network applications, a strong theoretical
foundation now exists for implementers.

\bibliographystyle{splncs03}
\bibliography{references}

\end{document}